\documentclass[conference]{IEEEtran}
\IEEEoverridecommandlockouts %
\usepackage{comment}
\includecomment{itwfull}
\excludecomment{journalonly}

\usepackage{url}
\usepackage{cite}
\ifCLASSINFOpdf
  \usepackage[pdftex]{graphicx}
  \DeclareGraphicsExtensions{.pdf}
\else
  \usepackage[dvips]{graphicx}
  \DeclareGraphicsExtensions{.eps}

\fi

\usepackage{multirow}

\usepackage{tikz}
\usetikzlibrary{arrows,shapes,backgrounds,positioning}

\usepackage{pgfplots}
\usepackage{booktabs}

\usepackage{amsthm}
\usepackage[cmex10]{amsmath}
\usepackage{amssymb}
\newtheorem{theorem}{Theorem}

\newtheorem{lemma}[theorem]{Lemma}

\theoremstyle{definition}

\theoremstyle{remark}

\DeclareMathOperator{\E}{\mathbb{E}}
\DeclareMathOperator{\rank}{\text{rk}}

\newcommand{\ffield}{\mathbb{F}}
\newcommand{\bH}{\mathbf{H}}

\newcommand{\idx}{\mathbf{i}}
\newcommand{\revise}[1]{{#1}}

\begin{document}
\title{Coding for Network-Coded Slotted ALOHA}

\author{\IEEEauthorblockN{Shenghao Yang\IEEEauthorrefmark{1}, Yi
    Chen\IEEEauthorrefmark{2}, Soung Chang Liew\IEEEauthorrefmark{3}
    and Lizhao You\IEEEauthorrefmark{3}}
  \IEEEauthorblockA{\IEEEauthorrefmark{1}Institute of Network Coding,
    \IEEEauthorrefmark{3}Information Engineering, 
    The Chinese University of Hong Kong, Hong Kong, China}
  \IEEEauthorblockA{\IEEEauthorrefmark{2}Electronic Engineering, City
    University of Hong Kong, Hong Kong, China}%
  \thanks{\revise{This
    work was partially supported by the National Natural Science
    Foundation of China under Grant 61471215. This work was partially
    supported by the General Research Funds (Project No. 414812) and
    the AoE grant (Project No.\ AoE/E-02/08), established under the
    University Grant Committee of the Hong Kong Special Administrative
    Region, China.}}}

\maketitle

\begin{abstract}
  Slotted ALOHA can benefit from physical-layer network coding (PNC)
  by decoding one or multiple linear combinations of the packets
  simultaneously transmitted in a timeslot, forming a system of linear
  equations. Different systems of linear equations are recovered in
  different timeslots.  A message decoder then recovers the original
  packets of all the users by jointly solving multiple systems of
  linear equations obtained over different timeslots. We propose the
  batched BP decoding algorithm that combines belief propagation (BP)
  and local Gaussian elimination.  Compared with pure Gaussian
  elimination decoding, our algorithm reduces the decoding complexity
  from cubic to linear function of the number of users. Compared with
  the ordinary BP decoding algorithm for low-density generator-matrix
  codes, our algorithm has better performance and the same order of
  computational complexity.  We analyze the performance of the batched
  BP decoding algorithm by generalizing the tree-based approach and
  provide an approach to optimize the system performance. 
\end{abstract}

\section{Introduction}

In a wireless multiple-access network operated with the slotted ALOHA
access protocol, a number of users transmit messages to a sink node
through a common wireless medium. Time is divided into discrete slots
and all transmissions start at the beginning of a timeslot.
Collisions/interferences occur when more than one user transmits in
the same timeslot \cite{roberts1975aloha}.  Successive interference
cancellation (SIC) was proposed for slotted ALOHA to resolve
collisions so that signals contained in collisions can be leveraged to
increase throughput \cite{casini2007contention}. In this approach, a
number of timeslots are grouped together as a frame. Each user aims to
deliver at most one packet per frame, but it can transmit 
copies of the same packet in different timeslots of the frame.

To see the essence, suppose we have two users and the first user transmits two
copies of packet $v_1$ in timeslots $1$ and $2$ respectively, and the
second user transmits one copy of packet $v_2$ in timeslots $2$. Since
no collision occurs, $v_1$ can be correctly decoded by the sink node
in timeslot $1$. A collision occurs in timeslot 2. In SIC, the sink
node can use $v_1$ decoded from timeslot 1 to cancel the interference
in timeslot 2. This approach can be applied iteratively to cancel more
interference, in a manner similar to the belief propagation (BP) decoding of LT
codes over erasure channels.  Slotted ALOHA with SIC has been
extensively studied based on the AND-OR-tree analysis, and optimal
designs have been obtained
\cite{liva2011graph,paolini2011high,stefanovic2012frameless,narayanan2012iterative,ghanbarinejad2013irregular}. %

\emph{Physical-layer network coding (PNC)} \cite{Zhang06} (also known
as \emph{compute-and-forward} \cite{Nazer11}) is recently applied to
wireless multiple-access network to improve the throughput
\cite{lu14n,ncma14,cocco2014}. Such multiple-access schemes, called
\emph{network-coded multiple access (NCMA)}, employ both PNC and
multiuser decoders at the physical layer to decode one or multiple
linear combinations of the packets simultaneously transmitted in a
timeslot.  Specifically, Lu, You and Liew \cite{lu14n} demonstrated by
a prototype that a PNC decoder may sometimes successfully recover linear
combinations of the packets when the traditional multiuser decoder (MUD)
\cite{verdu1998multiuser} that does not make use of PNC fails.
In the existing works on PNC (or compute-and-forward), the decoding of
the XOR of the packets of two users has been extensively investigated
\cite{shengli09,wubben10} (see also the overview
\cite{liew2013physical}). The decoding of multiple linear combinations
over a larger alphabet has been studied in \cite{Nazer11,
  feng2013algebraic,zhu2014isit}. %

In this paper, we consider slotted ALOHA employing PNC (and MUD) at
the physical layer, called \emph{network-coded slotted ALOHA (NCSA)}.
We assume that the physical-layer decoder at the sink node can
reliably recover one or multiple linear combinations of the packets
transmitted simultaneously in one timeslot.  Our work in this paper
does not depend on a specific PNC scheme. 
Specifically, we consider a $K$-user NCSA system, where each user has
one input packet to be delivered over a frame of timeslots. A
packet is the smallest transmission unit, which cannot be further
separated into multiple smaller transmission units. But it is allowed
to send multiple copies of a packet in different timeslots.  The
number of copies, called the \emph{degree}, is independently sampled
from a \emph{degree distribution}.  The linear equations decoded by
the physical layer in a timeslot form a system of linear equations.
Different systems of linear equations are recovered in different
timeslots. To recover the input packets of users, a message decoder is
then required to jointly solve these systems of linear equations
obtained over different timeslots. Though Gaussian elimination can be
applied to solve the input packets, the computational complexity is
$O(K^3+K^2T)$ finite-field operations, where $T$ the number of field
symbols in a packet. In this paper, we study the design of NCSA
employing an efficient message decoding algorithm.

With the possibility of decoding more than one linear combination of
packets in a timeslot, the coding problem induced by NCSA becomes
different from that of slotted ALOHA with SIC. We will show by an
example that the ordinary BP decoding algorithm of LT codes over
erasure channels is not optimal for NCSA. We instead propose a
\emph{batched BP decoding algorithm} for NCSA, where Gaussian
elimination is applied locally to solve the linear system associated
with each timeslot, and BP is applied between the linear systems
obtained over different timeslots. The computational complexity of our
algorithm is $O(KT)$ finite-field operations, which is of the same
order as the ordinary BP decoding algorithm. We analyze the asymptotic
performance of the batched BP decoding algorithm when $K$ is large by
generalizing the tree-based approach in \cite{luby98}. We provide an
approach to optimize the degree distribution based on our analytical
results.  

Though the batched BP decoding is similar to the one proposed for NCMA
\cite{yang14l,yang14c}, we cannot apply the analysis therein. In NCMA,
we assume that the number of users is fixed but the number of packets
to be delivered by each user tends to infinity. In NCSA, each user has
only one packet while the number of users can be large.

\revise{Similar schemes have been developed for random linear network
  coding over finite fields without explicitly considering the
  physical-layer effect, e.g., BATS codes and chunked codes (see
  \cite{yang14bats, yang14d} and the references therein). Here the
  technique for NCSA is different from BATS (or chunked) codes in two
  aspects. First, in BATS codes the degree distribution of batches is
  the parameter to be optimized, while in NCSA the degree distribution of
  the input packets (variable nodes) is the parameter to be
  optimized. Second, the decoding of BATS codes only solves
  the associated linear system of a batch when it is uniquely solvable (and hence
  recovers all the input packets involved in a batch), while the
  decoding of NCSA processes the associated linear system of a batch even
  when it is not uniquely solvable.}

\begin{itwfull}
  In the remainder of this paper, Section~\ref{sec:netw-coded-slott} formally introduces NCSA and
presents our main analytical result (Theorem~\ref{the:1}). An outline
of the proof of the theorem is given in
Section~\ref{sec:performance-analysis}. An example is provided in
Section~\ref{sec:example} to demonstrate the degree distribution
optimization and the numerical results. 
\end{itwfull}

\section{Network-Coded Slotted ALOHA}
\label{sec:netw-coded-slott}

In this section, we introduce the model of network-coded slotted ALOHA
(NCSA), the message decoding algorithm and the performance analysis
results.

\subsection{Slotted Transmission}

Fix a \emph{base field} $\mathbb{F}_q$ with $q$ elements and an
integer $m>0$.  Consider a wireless multiple-access network where $K$
source nodes (users) deliver information to a sink node through a
common wireless channel. Each user has one input packet for
transmission, formulated as a column vector of $T$ symbols in
the extension field $\mathbb{F}_{q^m}$. %

All the users are synchronized to a \emph{frame} consisting
of $n$ timeslots of the same duration.  The transmission of a packet
starts at the beginning of a timeslot, and the timeslots are long
enough for completing the transmission of a packet.  Each user transmits
a number of copies of its input packet within the frame. The
number of copies transmitted by a user, called the degree of the
packet, is picked independently according to a \emph{degree
  distribution} 
$\Lambda=(\Lambda_1,\ldots,\Lambda_{D})$, where $D$ is the maximum
degree. That is, with probability $\Lambda_d$, a user transmits $d$
copies of its input packet in $d$ different timeslots chosen uniformly
at random in the frame. Let $\bar \Lambda = \sum_{i=1}^Di\Lambda_i$, 
 $\Lambda(x) = \sum_{i} \Lambda_ix^i$ and $\Lambda'(x)=\sum_{i}i\Lambda_ix^{i-1}$. We also call $\Lambda(x)$ a
degree distribution.

Denote by $v_i$ the input packet of the $i$-th user.  Fix a
timeslot. Let $\Theta$ be the set of indices of the users who transmit a
packet in this timeslot.  The elements in $\Theta$ are ordered by the
natural order of integers. 
We assume that certain PNC scheme is applied, so
that the physical-layer decoder of the sink node can decode multiple
\emph{output packets}, each being a linear combination of $v_s,
s\in\Theta$ with coefficients over the base field $\ffield_q$.
Suppose that $B$ output packets are decoded ($B$ may vary from
timeslot to timeslot).  The collection of $B$ linear combinations can be expressed as
\begin{equation}\label{eq:batch}
  [u_1,\ldots,u_B] = [v_s,s\in\Theta] \bH,
\end{equation}
where $\bH$ is a $|\Theta|\times B$ full-column-rank matrix over
$\ffield_q$, called the \emph{transfer matrix}, and $[v_s,s\in\Theta]$
is the matrix formed by juxtaposing the vectors $v_s$, where $v_{s'}$
comes before $v_{s''}$ whenever $s'<s''$.

Note that in \eqref{eq:batch}, the algebraic operations are over the
field $\ffield_{q^m}$.  We call the set of packets
$\{u_1,\ldots,u_B\}$ decoded in a timeslot a \emph{batch}.
The cardinality of $\Theta$ (the number of users transmitting in a
timeslot) is called the \emph{degree of the batch/timeslot}. 
\revise{
We call the ratio $K/n$ the \emph{design rate} of NCSA.}

\begin{lemma}\label{lemma:1}
  When $K/n \rightarrow R$ as $K\rightarrow \infty$, the degree of a
  timeslot converges to the Poisson distribution with parameter
  $\lambda = R\bar \Lambda$ as $K\rightarrow \infty$.
\end{lemma}
\begin{itwfull}
\begin{proof}
  This is a special case of Lemma~\ref{lemma:1a} to be proved later in
  this paper.
\end{proof}
\end{itwfull}

Denote by $\mathcal{H}_{d}$ the collection of all the
full-column-rank, $d$-row matrices over $\ffield_q$, \revise{where we assume
that the empty matrix, representing the case that nothing is decoded,
is an element of $\mathcal{H}_{d}$.}  For a timeslot of degree $d$, we
suppose that the transfer matrix of the batch is $\bH\in
\mathcal{H}_d$ with probability $g(\bH|d)$. Further, we consider all
the users are symmetric so that for any $d\times d$ permutation matrix
$\mathbf{P}$,
\begin{equation}\label{eq:4}
  g(\bH|d) = g(\mathbf{P}\bH|d).
\end{equation}
The transfer matrices of all timeslots are independently generated
given the degrees of the timeslots. Examples of the
distribution $g$ will be given in Section~\ref{sec:example}.

\revise{We say a rate $R$ is \emph{achievable} by the NCSA
  system if for any $\epsilon>0$ and all sufficiently large $n$, at
  least $n(R-\epsilon)$ input packets are decoded correctly from the
  receptions of the $n$ timeslots with probability at least
  $1-\epsilon$.}

\subsection{Belief Propagation Decoding}

For multiple access described above, the goal of the sink node is to
decode as many input packets as possible during a frame. From the
output packets of the $n$ timeslots decoded by the physical layer, the
original input packets can be recovered by solving the linear
equations \eqref{eq:batch} of all the timeslots jointly. Gaussian
elimination has a complexity $O(K^3+K^2T)$ finite-field operations
when $n = O(K)$, which makes the decoding less efficient when $K$ is
large. %

The output packets of all the timeslots collectively can be regarded
as a low-density generator matrix (LDGM) code.  Similar to decoding an
LT code, which is also a LDGM code, we can apply the (ordinary) BP
algorithm to decode the output packets. In each step of the BP
decoding algorithm, an output packet of degree one is found, the
corresponding input packet is decoded, and the decoded input packet is
substituted into the other output packets in which it is involved. The
decoding stops when there are no more output packets of degree one.
\revise{However, as we will show in the next example, the ordinary BP
  decoding cannot decode some types of batches efficiently. We can
  actually do better than the ordinary BP decoding with little
  increase of decoding complexity by exploiting the batch structure of
  the output packets.}

For example, consider a batch of two packets $u_1$ and $u_2$ formed by
\begin{equation}
  \label{eq:1}
  \begin{bmatrix}
    u_1 & u_2
  \end{bmatrix} 
  =
  \begin{bmatrix}
    v_1 & v_2 & v_3 & v_4
  \end{bmatrix}
  \begin{bmatrix}
    1 & 0 \\
    0 & 1 \\
    1 & 1 \\
    1 & 1 
  \end{bmatrix}.
\end{equation}
Suppose that we use the ordinary BP decoding algorithm, and when the
BP decoding stops, $v_1$ is recovered \revise{by processing other batches}, but
$v_2$, $v_3$ and $v_4$ are not recovered. However, if we allow the decoder to
solve the linear system \eqref{eq:1}, we can further recover
$v_2=u_2-u_1+v_1$. The example shows that the BP decoding performance
can be improved if the linear system associated with a timeslot can be
solved locally.

Motivated by the above example, we propose the \emph{batched BP decoder}
for the output packets of the physical layer of NCSA.  The decoder
includes multiple iterations.  In the $i$-th iteration of the
decoding, $i=1,2,\ldots$ all the batches are processed individually
by the following algorithm: Consider a batch given in
\eqref{eq:batch}. Let $S\subset \Theta$ be the set of indices $r$ such
that $v_r$ is decoded in the previous iterations. When $i=0$, $S =
\emptyset$.  %
Let $\idx_\Theta:\Theta\rightarrow \{1,\ldots,|\Theta|\}$ be the
one-to-one mapping preserving the order on $\Theta$, i.e.,
$\idx_\Theta(s_1)<\idx_\Theta(s_2)$ if and only if $s_1<s_2$.
We also write $\idx(s)$ when $\Theta$ is clear from the
context. 
The algorithm first substitutes the values of $v_r, r\in
S$ into \eqref{eq:batch} and obtain
\begin{equation}
  \label{eq:2}
  [u_1,\ldots,u_B] - [v_r, r\in S] \bH^{\idx[S]} = [v_s,s\in\Theta\setminus S] \bH^{\idx[\Theta\setminus S]},
\end{equation}
where $\bH^{\idx[S]}$ is the submatrix of $\bH$ formed by the rows
indexed by $\idx[S]$.  The algorithm then applies Gaussian
(Gauss-Jordan) elimination on the above linear system so that
$\bH^{\idx[\Theta\setminus S]}$ is transformed into the reduced column
echelon form $\tilde \bH$ and \eqref{eq:2} becomes
\begin{equation}
  \label{eq:3}
  [\tilde u_1,\ldots, \tilde u_B] = [v_s,s\in\Theta\setminus S] \tilde \bH.
\end{equation}
Suppose that the $j$-th column of $\tilde \bH$ has only one nonzero
component (which should be one) at the row corresponding to user
$s$. The value of $v_s$ is then $\tilde u_j$ and hene recovered.  The algorithm
returns the new recovered input packets by searching the columns of
$\tilde \bH$ with only one non-zero component.

For a batch with degree $d$, the complexity of the above decoding is
$O(d^3+d^2T)$. Suppose that $K/n$ is a constant and the maximum degree
$D$ does not change with $K$. Since the degree of a
batch converge to the Poisson distribution with parameter
$\frac{K}{n}\bar\Lambda$ (see Lemma~\ref{lemma:1}), the average
complexity of decoding a batch is $O(T)$ finite-field
operations. Hence the total decoding complexity is $O(KT)$
finite-field operations.

\subsection{Decoding Performance}

For an integer $j$, denoted by $[j]$ the set of integers
$\{1,\ldots,j\}$. When $j\leq 0$, $[j]=\emptyset$. For any
$\mathbf{H}\in \mathcal{H}_{d}$, define $\gamma(\mathbf{H})$ as the
collection of all subsets $V$ of $[d-1]$ such that 
in the linear system \eqref{eq:batch},
$v_{\idx^{-1}(d)}$ can be uniquely solved when the values of $v_r, r\in \idx^{-1}[V]$ are known.
Taking the transfer matrix in \eqref{eq:1} as an example, we have
\begin{equation*}
  \gamma(\mathbf{H}) = \{ \{1, 3\}, \{2,3\}, \{1,2,3\}\}.
\end{equation*}
For a timeslot of degree one, the transfer matrix $\bH$ is the
one-by-one matrix with the unity. Then
$\gamma(\bH) = \{ \emptyset\}$.
For any intiger $k\geq 0$, define
\begin{equation*}
  \Gamma_k(x) = \sum_{\bH\in \mathcal{H}_{k+1}} g(\bH|k+1) \sum_{S\in
    \gamma(\bH)} x^{|S|}(1-x)^{k-|S|}.
\end{equation*}
\revise{In other words, $\Gamma_k(x)$ is the probability that when $k+1$ users
transmitted in a timeslot, the input packets of the user with the
largest index can be recovered if each of the other users' packet is
known with probability $x$.}

We assume that the maximum degree $D$ is a contant that does not
change with $K$. The following theorem, proved in the next section,
tells us the decoding performance of $l$ iterations of the batched BP
decoder when $K$ is sufficiently large. We apply the convention that
$0^0=1$. 

\begin{theorem}\label{the:1}
  Fix real numbers $R>0$, $\epsilon>0$ and an integer $l>0$.  Consider a
  multiple-access system described above with $K$ users and $n =
  \lceil K/R\rceil$ timeslots. 
  Define 
  \begin{equation*}
    z^*_l = 1 - \Lambda\left(1-\sum_k\frac{\lambda^k e^{-\lambda}}{k!}\Gamma_k(z_{l-1})\right),
  \end{equation*}
  where   $z_0=0$ and   for $1\leq i <l$
  \begin{equation*}
    z_i = 1 - \Lambda'\left(1-\sum_k\frac{\lambda^k e^{-\lambda}}{k!}\Gamma_k(z_{i-1})\right)/\bar\Lambda,
  \end{equation*}
  where $\lambda = R \bar\Lambda$.  Then for any sufficiently large
  $K$, $l$ iterations of the batch BP decoder will recover at
  least $K(z^*_l-\epsilon)$ input packets with
  probability at least $1-\exp(-c\epsilon^2K)$, where $c$ is a number
  independent of $K$ and $n$.
\end{theorem}
\begin{proof}
  See Section \ref{sec:performance-analysis}.
\end{proof}

\begin{lemma}\label{lemma:gamma}
  $\Gamma_k(x)$ is an increasing function of $x$.
\end{lemma}
\begin{itwfull}
\begin{proof}
  This lemma can be proved by applying \cite[Lemma~13]{yang14c}.  
\end{proof}
\end{itwfull}

\subsection{Degree Distribution Optimization}
\label{sec:perf-eval}

Theorem~\ref{the:1} induces a general approach to optimize the degree
distribution $\Lambda$. Let
\begin{equation*}
  f(x;\lambda) = 1 - \Lambda'\left(1-\sum_k\frac{\lambda^k
      e^{-\lambda}}{k!}\Gamma_k(x)\right)/\bar\Lambda. 
\end{equation*}
We have $z_i = f(z_{i-1};\lambda), i=1,\ldots,l-1$. Suppose that we allow
$l\rightarrow \infty$. The sequence $\{z_i\}$ is increasing (implied
by Lemma~\ref{lemma:gamma}) and
converges to the first value $x>0$ such that $f(x;\lambda)=x$. 
For given value of $\lambda$, $0< \epsilon<1$
and $0<\eta\leq 1$, we can optimize the degree distribution $\Lambda$
by solving 
\begin{equation}\label{eq:6}
  \begin{IEEEeqnarraybox}[][c]{r.l}
    \max & R \\
    \text{s.t.} & 
    f(x;\lambda) \geq x(1+\epsilon), \quad \forall x\in (0,\eta], \\
    & \sum_ii\Lambda_i = \lambda/R, \sum_i\Lambda_i=1, \Lambda_i\geq 0.
  \end{IEEEeqnarraybox}
\end{equation}

\begin{theorem}\label{the:2l}
  Denote by $R(\lambda,\epsilon)$ the optimal value of the above
  optimization. Then the rate
  \begin{equation*}
    R^*(\lambda,\epsilon) = R(\lambda,\epsilon) 
  \left(1-\Lambda\left(1-\sum_{k}\frac{\lambda^k e^{-\lambda}}{k!}
  \Gamma_k(\eta)\right)\right)
  \end{equation*}
  packet per timeslot is achievable for the batched BP decoding algorithm.
\end{theorem}
\begin{itwfull}
\begin{proof}
  For any $\delta >0$, let $R = R(\lambda,\epsilon) - \sqrt{\delta}$. We show
that for sufficiently large $K$, there exists a degree distribution
$\Lambda$ such that using $n\leq K/R$ timeslots, the batch BP decoding
algorithm can recover at least $K(\eta^*-\sqrt{\delta})$ input
packets with high probability, where
\begin{equation*}
  \eta^* = \left(1-\Lambda\left(1-\sum_{k}\frac{\lambda^k e^{-\lambda}}{k!}
  \Gamma_k(\eta)\right)\right).
\end{equation*}
That is the code has a rate at least $R^*(\lambda,\epsilon)-\delta$
packet per timeslot. 

Let $n=\lceil K/R(\lambda,\epsilon)\rceil$. For the degree
distribution $\Lambda$ achieving $R(\lambda,\epsilon)$ in
\eqref{eq:6}, we know by Theorem~\ref{the:1} that at
$K(z_l^*-\sqrt{\delta})$ input packets can be recovered with high
probability. We know that the sequence $\{z_i\}$ converges to a value
larger than $\eta$. Then there exists a sufficiently large $l$ such
that $z_{l-1}\geq \eta$. Thus, $z_l^* \geq \eta^*$. The proof is completed.
\end{proof}
\end{itwfull}

\section{Performance Analysis}
\label{sec:performance-analysis}

We generalize the tree-based approach \cite{luby98} to analyze the
performance of the batched BP decoder and prove
Theorem~\ref{the:1}. %

\subsection{Decoding Graph}

The relation between the input packets and the timeslots can be
represented by a random Tanner graph $G$, where the input packets are
represented by the variable nodes, and timeslots are represented by the
check nodes.  We henceforth equate a variable node with the
corresponding input packet, and a check node with the corresponding
timeslot. There exists an edge between a variable node and a check
node if and only if the corresponding input packet is transmitted in
the timeslot. Associated with each check node is a random transfer
matrix $H$. For given degree $d$ of the timeslot, the distribution of
$H$ is $g(\cdot|d)$.

The $l$-neighborhood of a variable node $v$, denoted by $G_l(v)$, is
the subgraph of $G$ that includes all the nodes with distance less
than or equal to $l$ from variable node $v$, as well as all the edges
involved. Since $G_l(v)$ has the same distribution for all variable
node $v$, we denote by $G_l$ the generic random graph with the same
distribution as $G_l(v)$. After $l$ iterations of the batched BP
decoding, whether or not a variable node $v$ is decoded is determined
by its $2l$-neighborhood.  

\begin{itwfull}
  Motived by the tree-based approach, in the remainder of this section,
we first analyze the decodable probability of the root node of a
random tree, and then show that the decoding performance of $G_{2l}$
is similar to that of the tree. The proof of Theorem~\ref{the:1} is
then completed by a martingale argument.

\end{itwfull}

\subsection{Tree Analysis}

Fix two degree distributions $\alpha(x)$ and $\beta(x)$.
Let $T_l$ be a tree of $l+1$ levels. The root of the tree is at level
$0$ and the leaves are at level $l$. Each node at an even level is a
variable node, and each node at an odd level is a check node. 
The probability that the root node has $i$ children is $\Lambda_i$. 
Except for the root node, all the other variable nodes have $i$
children with probability $\alpha_i$. All the check node has $i$
children with probability $\beta_i$. 
\begin{itwfull}
  An instance of $T_4$ is shown in Fig.~\ref{fig:tree}.
\end{itwfull}

\begin{lemma}
  Let $x_l^*$ be the probability that the root variable node is
  decodable by applying the batched BP decoding on $T_{2l}$.  
  We have
  \begin{equation*}
    x_l^* = 1 - \Lambda\left(1-\textstyle \sum_{k} \beta_k \Gamma_k(x_{l-1}) \right),
  \end{equation*}
  where $x_0=0$ and for $1 \leq i < l$,
  \begin{equation*}
    x_i = 1 - \alpha\left(1- 
    \textstyle \sum_{k} \beta_k \Gamma_k(x_{i-1}) \right).
  \end{equation*}
\end{lemma}
\begin{itwfull}
\begin{proof}
  Denote by $y_i$ the probability that a check node at level
  $2(l-i)+1$ can recover its parent variable node by solving the
  associated linear system of this check node with possibly the
  knowledge of its children variable nodes. We have $x^*_l =
  1-\Lambda(1-y_l)$. Suppose that a variable node at level $2(l-i)$ is
  decodable by at least one of its children check node with
  probability $\hat{x}_{i}$, $0\leq i < l$.  We have $\hat{x}_i = 1 -
  \alpha(1-y_{i})$ for $0<i<l$ and $\hat{x}_0=0$.
  
  Fix a check node $c$ at level $2(l-i)+1$. With probability
  $g(\bH|k+1)\beta_k$, the check node has $k$ children and the
  associated linear system has $\bH$ as the transfer matrix. We
  permutate the rows of $\bH$ such that the \emph{last} row of $\bH$
  corresponds to the parent variable node. By \eqref{eq:4}, the
  permutation does not change the distribution $g(\bH|k+1)$.  Index
  the $k$ children by $1, \ldots, k$.  Using Gaussian elimination in
  the batched BP decoder, the parent variable node of check node $c$
  can be recovered if and only if for certain $S \in
  \gamma(\bH)$, all the children variable nodes indices by $S$
  are decodable. Therefore, the probability that the parent variable
  node of $c$ is decodable is $\sum_{S\in \gamma(\bH)}
  \hat{x}_{i-1}^{|S|}(1-\hat{x}_{i-1})^{k-|S|}$ for transfer matrix
  $\bH$. Considering all the possible transfer matrices, we have $y_i
  = \sum_k\beta_k\Gamma_k(\hat{x}_{i-1})$. The proof is completed by
  $x_i=\hat{x}_i$.
\end{proof}
\end{itwfull}

\begin{itwfull}
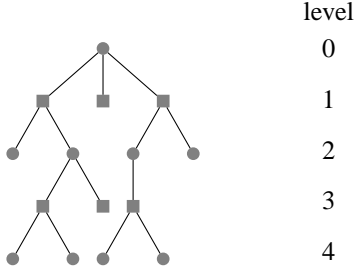
\begin{figure}
  \centering
  \begin{tikzpicture}[
    vnode/.style={circle,draw=gray,fill=gray,thick,inner
 sep=0pt,minimum size=4pt},
    cnode/.style={draw=gray,fill=gray,thick,inner
 sep=0pt,minimum size=4pt},
    level distance=7mm,level/.style={sibling distance=8mm}]

  \node[vnode] at (0,0) {}
  child { node[cnode] {}
    child {node[vnode]{}}
    child {node[vnode]{} 
      child {node[cnode]{}
        child {node[vnode]{}}
        child {node[vnode]{}}
      }
      child {node[cnode]{}}
    }
  }
  child { node[cnode] {}}
  child { node[cnode] {}
    child {node[vnode]{} 
      child {node[cnode]{}
        child {node[vnode]{}}
        child {node[vnode]{}}
      }
    }
    child {node[vnode]{} }
  };

  \matrix[column sep=8mm,row sep=2mm] at (3,-1.1) {
    \node[label=above:level] {0}; \\
    \node {1}; \\
    \node {2}; \\
    \node {3}; \\
    \node {4}; \\
  };

  \end{tikzpicture}
  \caption{An instance of $T_4$.}
  \label{fig:tree}
\end{figure}
\end{itwfull}

\begin{itwfull}
We prove the following stronger result than  Lemma~\ref{lemma:1}.
\begin{lemma}\label{lemma:1a}
  Suppose that $K/n \rightarrow R$ as $K\rightarrow \infty$.
  Fix a timeslot $t$ and an integer $k\geq 0$.
  Under the condition that a fixed set of $k$ users do not transmit 
  at timeslot $t$, the degree of 
  timeslot $t$ converges to the Poisson distribution with parameter
  $\lambda = R\bar \Lambda$ as $K\rightarrow \infty$. 
\end{lemma}
\begin{IEEEproof}
  Let $\Theta$ be the set of users that do not transmit at timeslot
  $t$. For each user that is not in $\Theta$, the probability that
  this user transmits a packet at timeslot $t$ is $\bar\Lambda/n$,
  when $n$ is larger than $D$. Therefore, the degree of timeslot $t$
  follows a binomial distribution with parameter $(K-k,
  \bar\Lambda/n)$, which converges to the Poisson distribution with
  parameter $R\bar\Lambda$ when $K\rightarrow \infty$.
\end{IEEEproof}
\end{itwfull}

For a positive integer $L$, let
$\epsilon_L = 1 - \sum_{d=0}^L \frac{\lambda^de^{-\lambda}}{d!}.$
We are interested in the following instances of $\alpha$ and $\beta$
\begin{IEEEeqnarray}{r.r}
  \label{eq:ab}
  \alpha(x) = \frac{\Lambda'(x)}{\bar\Lambda},
  &
  \beta(x) = 
  \frac{1}{1-\epsilon_L}
  \sum_{k=0}^L \frac{\lambda^k e^{-\lambda}}{k!} x^k.
\end{IEEEeqnarray}
Let $\mathcal{G}_l(L)$ be the set of trees of $l+1$ levels where each
check node has at most $L$ children and each variable node has at most
$D$ children. 
\begin{lemma}\label{lemma:4}
  When $K$ is sufficiently
  large, for any $\mathbf{G}_l\in\mathcal{G}_l(L)$,
  \begin{equation*}
    \Pr\{G_l = \mathbf{G}_l\} \geq \Pr\{T_l = \mathbf{G}_l\} - c_{l,L}\epsilon_L,
  \end{equation*}
  where $c_{l,L} = O(L^{\lfloor l/2\rfloor})$ and the degree distributions
  of $T_l$ are given in \eqref{eq:ab}.
\end{lemma}
\begin{itwfull}
\begin{proof}
  We show by induction that
\begin{equation}
  \label{eq:5}
  \Pr\{G_l = \mathbf{G}_l\} \geq \Pr\{T_l =
\mathbf{G}_l\} - c_{l,L}\epsilon_L,
\end{equation}
where $c_{l,L} = O(L^{\lfloor l/2\rfloor})$.

We prove the lemma by induction. When $l=1$, $G_1$ and $T_1$ follow
the same distribution. For $l>1$, we have
\begin{equation*}
  \Pr\{G_l = \mathbf{G}_l\} = \Pr\{G_l = \mathbf{G}_l| G_{l-1} =
  \mathbf{G}_{l-1}\} \Pr\{G_{l-1} = \mathbf{G}_{l-1}\},
\end{equation*}
where $\mathbf{G}_{l-1}$ is the subgraph of $\mathbf{G}_l$ obtained by
removing the leaf nodes. 
We assume that 
\begin{equation*}
  \Pr\{G_{l-1} = \mathbf{G}_{l-1}\} \geq \Pr\{T_{l-1} =
\mathbf{G}_{l-1}\} - c_{l-1,L}\epsilon_L,
\end{equation*}
for certain function $c_{l-1,L} = O(L^{\lfloor (l-1)/2 \rfloor})$.
We then prove \eqref{eq:5} with $l>0$ for two cases:$l$ is even and
$l$ is odd.

We first consider the case that $l$ is even. Suppose that
$\mathbf{G}_{l-1}$ has $N$ leaf \emph{check} nodes, which are at level
$l-1$ of $\mathbf{G}_l$. Denote by $k_i$ the number of children
variable nodes of the $i$-th check node at level $l-1$ in
$\mathbf{G}_l$. Since $\mathbf{G}_l \in \mathcal{G}_l(L)$, we have
$k_i \leq L$. By Lemma~\ref{lemma:1a}, we have
\begin{IEEEeqnarray*}{rCl}
  \Pr\{G_l = \mathbf{G}_l| G_{l-1} =
  \mathbf{G}_{l-1}\} \rightarrow \prod_{i=1}^N
  \frac{\lambda^{k_i}e^{-\lambda}}{k_i!},\quad K\rightarrow\infty.
\end{IEEEeqnarray*}
On the other hand, we have
\begin{IEEEeqnarray*}{rCl}
  \Pr\{T_l = \mathbf{G}_l| T_{l-1} = 
  \mathbf{G}_{l-1}\} = \frac{1}{(1-\epsilon_L)^N}\prod_{i=1}^N
  \frac{\lambda^{k_i}e^{-\lambda}}{k_i!}.
\end{IEEEeqnarray*}
Therefore, for sufficiently large $K$,
\begin{IEEEeqnarray*}{rCl}
  \IEEEeqnarraymulticol{3}{l}{ \Pr\{G_l = \mathbf{G}_l| G_{l-1} =
  \mathbf{G}_{l-1}\} - \Pr\{T_l = \mathbf{G}_l| T_{l-1} = 
  \mathbf{G}_{l-1}\}} \\
  & \geq & (1-\epsilon_L)^N- 1 - \epsilon_L \\
  & \geq & -(N+1)\epsilon_L.
\end{IEEEeqnarray*}
Note that $N= O(L^{\lfloor l/2 \rfloor})$.

We then consider the case that $l$ is odd. Suppose that
$\mathbf{G}_{l-1}$ has $N$ leaf \emph{variable} nodes, which are at
level $l-1$ of $\mathbf{G}_l$. Denote by $k_i$ the number of children
check nodes of the $i$-th variable node at level $l-1$ of
$\mathbf{G}_l$. We know that $k_i \leq D-1$. We then have
\begin{IEEEeqnarray*}{rCl}
  \IEEEeqnarraymulticol{3}{l}{\Pr\{G_l = \mathbf{G}_l| G_{l-1} =
  \mathbf{G}_{l-1}\}} \\
  & \rightarrow & \prod_{i=1}^N \frac{(k_i+1)\Lambda_{k_i+1}}{\sum_d
    d\Lambda_{d}} \\
  & = & 
  \Pr\{T_l = \mathbf{G}_l| T_{l-1} =  \mathbf{G}_{l-1}\}.
\end{IEEEeqnarray*}
Therefore, for sufficiently large $K$, 
$\Pr\{G_l = \mathbf{G}_l| G_{l-1} =
  \mathbf{G}_{l-1}\} \geq \Pr\{T_l = \mathbf{G}_l| T_{l-1} = 
  \mathbf{G}_{l-1}\} -  \epsilon_L$.
\end{proof}
\end{itwfull}

\subsection{Proof of Theorem~\ref{the:1}}

Now we are ready to prove Theorem~\ref{the:1}. We say $G_l$ or $T_l$
is decodable if its root is decodable by the batched BP decoding
algorithm. Fix a sufficiently large $L$. We have
\begin{IEEEeqnarray*}{rCl}
  \IEEEeqnarraymulticol{3}{l}{\Pr\{G_{2l}\in \mathcal{G}_{2l}(L) \text{ and is decodable}\}} \\ & \geq &
  \sum_{\mathbf{G}\in\mathcal{G}_{2l}(L)} \Pr\{\mathbf{G} \text{ is decodable}\}  \Pr\{G_{2l} = \mathbf{G}\} \\
  & \geq & 
  \sum_{\mathbf{G}\in\mathcal{G}_{2l}(L)} \Pr\{\mathbf{G} \text{ is decodable}\}
  (\Pr\{{T}_{2l} = \mathbf{G}\} - \frac{\epsilon}{4|\mathcal{G}_{2l}(L)|}) \\
  & \geq & 
  \Pr\{{T}_{2l} \text{ is decodable}\} - \epsilon/4 
   =    x^*_l - \epsilon/4  
   \geq    z^*_l - \epsilon/2,
\end{IEEEeqnarray*}
where the second inequality follows from Lemma~\ref{lemma:4} and the
last inequality follows that $x_l^*\rightarrow z_l^*$ when
  $L\rightarrow \infty$.

Let $A$ be the number of variable nodes $v$ with $G_{2l}(v) \in \mathcal{G}_{2l}(L)$ and decodable. 
We have $\E[A] \geq (z^*_l-\epsilon/2)K$. 
For $i=1,\ldots, K$, denote $Z_i = G_1(v_i)$.
Define $X_i = \E[A|Z_1,\ldots,Z_i]$. By definition, $X_i$ is a Doob's
martingale with $X_0=\E[A]$ and $X_K=A$. Since the exposure of a variable
node will affect the degrees of a constant number of subgraphs
$G_{2l}(v)$ with check node degree $\leq L+1$, we have
$|X_i-X_{i-1}| \leq c'$, a constant does not depend on $K$. Applying
the Azuma-Hoeffding Inequality, we have
  \begin{equation*}
    \Pr\{A \leq \E[A] - \epsilon/2 K \} \leq \exp\left(- \frac{\epsilon^2K}{8c'^2} \right).
  \end{equation*}
  Hence $\Pr\{A > (z^*_l-\epsilon) K\} > 1 - \exp\left(-
    \frac{\epsilon^2K}{8c'^2} \right)$. This completes the proof of
  Theorem~\ref{the:1}.

\section{An Example}
\label{sec:example}

In this section, we use an example to illustrate how the proposed NCSA
scheme works. 
Here $q=2$ and $m=1$. 
Fix an integer $N\geq 2$. We consider the PNC scheme that
has the following outputs:
i) When one user transmits in a timeslot, the
  packet of the user is decoded;
ii) When two to $N$ users transmit in a timeslot, one or two binary
  linear combinations of the input packets are decoded; and 
iii) When more than $N$ users transmit in a timeslot, nothing is
  decoded.

Taking $N=3$ as an example, when one user transmits in a timeslot, the transfer
matrix is $\bH_{1} = [1]$, and $g(\bH_1|1)=1$.
When two users transmit in a timeslot, the 
possible transfer matrices are
\begin{IEEEeqnarray*}{c}
  \bH_{21}=
  \begin{bmatrix}
    1 \\ 1
  \end{bmatrix},
  \bH_{22}=
  \begin{bmatrix}
    1 & \\ & 1
  \end{bmatrix}, 
  \bH_{23}=
  \begin{bmatrix}
    & 1 \\ 1 &
  \end{bmatrix}.
\end{IEEEeqnarray*}
Since $\bH_{22}$ and $\bH_{23}$ have the same probability,
$g(\bH_{21}|2) + 2g(\bH_{22}|2) = 1.$
Now consider that three users transmit in a timeslot.
Define
\begin{equation*}
   \bH_{31}  = 
  \begin{bmatrix}
    1 \\ 1 \\ 1
  \end{bmatrix}, 
  \bH_{32}  = 
  \begin{bmatrix}
    1 & \\ 1 & \\ & 1
  \end{bmatrix}, 
  \bH_{33} = 
  \begin{bmatrix}
    1 & \\ 1 & 1 \\ & 1
  \end{bmatrix}.
\end{equation*}
The possible transfer matrices are given by the row permutations of
$\bH_{3i}$, $i=1,2,3$. Note that for two transfer matrices that are
permutation of each other, they have the same probability to
occur. Thus we have
\begin{equation*}
  g(\bH_{31}|3) + 3g(\bH_{32}|3) + 6g(\bH_{33}|3) = 1.
\end{equation*}
\begin{itwfull}
  We then have
\begin{IEEEeqnarray*}{rCl}
  \Gamma_0(x) & = & 1, \quad
  \Gamma_1(x)  =  g(\bH_{21}|2)x+2g(\bH_{22}|2),\\
  \Gamma_2(x) & = & g(\bH_{31}|3)x^2+g(\bH_{32}|3)(1+2x) \\
  & & +g(\bH_{33}|3)(8x-2x^2).
\end{IEEEeqnarray*}

\end{itwfull}

  In general, for a timeslot of $d$ users, we denote by $\bH_{d1}$ the
  single column transfer matrix of all ones. For transfer matrices of
  two columns, there are three types of rows: $[0,1]$, $[1,0]$ and
  $[1,1]$. Denote by $\bH_{d2}(a)$ a generic transfer matrix
  with $a$ rows of type $[0,1]$ and $d-a$ rows of type $[1,0]$. 
  Here $0<a \leq \lfloor d/2 \rfloor$. All the row permutations of $\bH_{d2}(a)$ are
  possible transfer matrices. 

  Denote by $\bH_{d3}(a_1,a_2)$ a generic transfer matrix with $a_1$
  rows of type $[0,1]$, $a_2$ rows of type $[1,0]$ and $d-a_1-a_2$
  rows of type $[1,1]$. Here $a_2\geq a_1>0$ and $a_1+a_2<d$. 
  All the row permutations of $\bH_{d3}(a_1,a_2)$ are
  possible transfer matrices. Thus,
  \begin{IEEEeqnarray*}{rCl}
    1 & = & g(\bH_{d1}|d) + \sum_{a=1}^{\lfloor d/2 \rfloor}\binom{d}{a}g(\bH_{d2}(a)|d) \\ & &
    + \sum_{a_1=1}^{d-2}\sum_{a_2=a_1}^{d-a_1-1} \binom{d}{a_1,a_2}
    g(\bH_{d3}(a_1,a_2)|d).
  \end{IEEEeqnarray*}
\begin{itwfull}
  We can then calculate that
  \begin{IEEEeqnarray*}{rCl}
    \IEEEeqnarraymulticol{3}{l}{\Gamma_{d-1}(x) = g(\bH_{d1}|d) x^{d-1}} \\
    & + &
    \sum_{a=1}^{\lfloor d/2 \rfloor}g(\bH_{d2}(a)|d)\left[\binom{d-1}{a-1}x^{a-1} + \binom{d-1}{a}x^{d-a-1} \right] \\ 
    & 
    + & \sum_{a_1=1}^{d-2}\sum_{a_2=a_1}^{d-a_1-1}
    g(\bH_{d3}(a_1,a_2)|d) \Bigg[\binom{d-1}{a_1-1,a_2} x^{d-a_2-1}
    \\
    & + &
    \binom{d-1}{a_1,a_2-1} x^{d-a_1-1} \\ 
    & + & \binom{d-1}{
      a_1,a_2} x^{d-a_1-a_2-1}(x^{a_1}+x^{a_2}-x^{a_1+a_2}) \Bigg].
  \end{IEEEeqnarray*}
\end{itwfull}

Given average degree $\lambda$ of a timeslot, the
average number of output packets decoded in a timeslot converges to
\begin{equation*}
  U(\lambda) = \sum_d \frac{\lambda^d e^{-\lambda}}{d!}
\sum_{\bH\in \mathcal{H}_d} \rank(\bH) g(\bH|d),
\end{equation*}
when $K\rightarrow\infty$. The achievable rate of the NCSA system is
upper bounded by $U(\lambda)$ packets per timeslot. In the case of
this example, the achievable rate bound is given by
\begin{IEEEeqnarray*}{rCl}
U(\lambda) & = & \sum_{d=1}^N \frac{\lambda^d e^{-\lambda}}{d!}
\Bigg[ g(\bH_{d1}|d) + 2\sum_{a=1}^{\lfloor d/2 \rfloor}\binom{d}{a}g(\bH_{d2}(a)|d) \\ & &
    + 2\sum_{a_1=1}^{d-2}\sum_{a_2=a_1}^{d-a_1-1} \binom{d}{a_1,a_2}
    g(\bH_{d3}(a_1,a_2)|d)
\Bigg].
\end{IEEEeqnarray*}
Note that the upper bound is in general not tight since the packets
decoded in different timeslots can be the same. 

We solve \eqref{eq:6} for the
above example with the results in Fig.~\ref{fig:1}, where we assume
the uniform distribution for each possible transfer matrices.
We also evaluate the corresponding upper bound $U(\lambda)$ for
comparison.

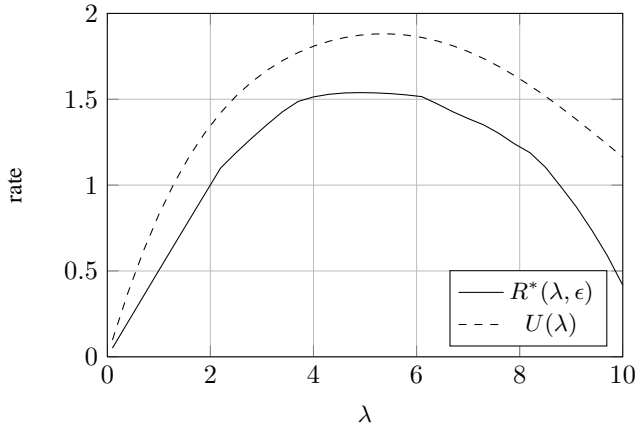
\begin{figure}
  \centering
  \begin{tikzpicture}
    \begin{axis}[
      xlabel=$\lambda$,ylabel=rate,
      width=240,height=175,
      xmin=0,xmax=10,
      ymin=0,ymax=2,
      legend pos= south east,
      legend style={font=\small},
      label style={font=\small},
      mark size={1.0},
      grid=both]

      \addplot[] table[x=lambda,y=rate] {ex2.txt};
      \addlegendentry{$R^*(\lambda,\epsilon)$}
      \addplot[dashed] table[x=lambda,y=bound] {ex2.txt};
      \addlegendentry{$U(\lambda)$}
    \end{axis}
  \end{tikzpicture}
  \caption{Achievable rates with $\eta=0.99$, $N=10$ and $\epsilon=0.001$.
    $R^*(\lambda,\epsilon)$ is the maximum achievable rate
    optimized w.r.t $\lambda$, and $U(\lambda)$ is an upper bound on the
    maximum achievable rate w.r.t. $\lambda$.}
  \label{fig:1}
\end{figure}


\begin{thebibliography}{10}
\providecommand{\url}[1]{#1}
\csname url@samestyle\endcsname
\providecommand{\newblock}{\relax}
\providecommand{\bibinfo}[2]{#2}
\providecommand{\BIBentrySTDinterwordspacing}{\spaceskip=0pt\relax}
\providecommand{\BIBentryALTinterwordstretchfactor}{4}
\providecommand{\BIBentryALTinterwordspacing}{\spaceskip=\fontdimen2\font plus
\BIBentryALTinterwordstretchfactor\fontdimen3\font minus
  \fontdimen4\font\relax}
\providecommand{\BIBforeignlanguage}[2]{{%
\expandafter\ifx\csname l@#1\endcsname\relax
\typeout{** WARNING: IEEEtran.bst: No hyphenation pattern has been}%
\typeout{** loaded for the language `#1'. Using the pattern for}%
\typeout{** the default language instead.}%
\else
\language=\csname l@#1\endcsname
\fi
#2}}
\providecommand{\BIBdecl}{\relax}
\BIBdecl

\bibitem{roberts1975aloha}
L.~G. Roberts, ``{ALOHA} packet system with and without slots and capture,''
  \emph{ACM SIGCOMM Computer Communication Review}, vol.~5, no.~2, pp. 28--42,
  1975.

\bibitem{casini2007contention}
E.~Casini, R.~De~Gaudenzi, and O.~R. Herrero, ``Contention resolution diversity
  slotted {ALOHA} ({CRDSA}): An enhanced random access schemefor satellite
  access packet networks,'' \emph{Wireless Communications, IEEE Transactions
  on}, vol.~6, no.~4, pp. 1408--1419, 2007.

\bibitem{liva2011graph}
G.~Liva, ``Graph-based analysis and optimization of contention resolution
  diversity slotted {ALOHA},'' \emph{Communications, IEEE Transactions on},
  vol.~59, no.~2, pp. 477--487, 2011.

\bibitem{paolini2011high}
E.~Paolini, G.~Liva, and M.~Chiani, ``High throughput random access via codes
  on graphs: Coded slotted {ALOHA},'' in \emph{Proc. IEEE ICC}, 2011.

\bibitem{stefanovic2012frameless}
C.~Stefanovic, P.~Popovski, and D.~Vukobratovic, ``Frameless {ALOHA} protocol
  for wireless networks,'' \emph{Communications Letters, IEEE}, vol.~16,
  no.~12, pp. 2087--2090, 2012.

\bibitem{narayanan2012iterative}
K.~R. Narayanan and H.~D. Pfister, ``Iterative collision resolution for slotted
  aloha: An optimal uncoordinated transmission policy,'' in
  \emph{Proc. IEEE ISTC}, 2012.

\bibitem{ghanbarinejad2013irregular}
M.~Ghanbarinejad and C.~Schlegel, ``Irregular repetition slotted aloha with
  multiuser detection,'' in \emph{Proc. WONS}.\hskip 1em plus 0.5em minus
  0.4em\relax IEEE, 2013.

\bibitem{Zhang06}
S.~Zhang, S.~C. Liew, and P.~P. Lam, ``Hot topic: Physical-layer network
  coding,'' in \emph{Proc. MobiCom '06}, New York, NY, USA, 2006.

\bibitem{Nazer11}
B.~Nazer and M.~Gastpar, ``Compute-and-forward: Harnessing interference through
  structured codes,'' \emph{Information Theory, IEEE Transactions on}, vol.~57,
  no.~10, pp. 6463--6486, 2011.

\bibitem{lu14n}
L.~Lu, L.~You, and S.~C. Liew, ``Network-coded multiple access,'' \emph{Mobile
  Computing, IEEE Transactions on}, vol.~13, no.~12, pp. 2853--2869, Dec 2014.

\bibitem{ncma14}
L.~You, S.~Liew, and L.~Lu, ``Network-coded multiple access ii : Toward
  realtime operation with improved performance,'' \emph{Selected Areas in
  Communications, IEEE Journal on}, vol.~PP, no.~99, pp. 1--1, 2014.

\bibitem{cocco2014}
G.~Cocco and S.~Pfletschinger, ``Seek and decode: Random multiple access with
  multiuser detection and physical-layer network coding,'' in
  \emph{Proc. IEEE ICC}, 2014.

\bibitem{verdu1998multiuser}
S.~Verdu, \emph{Multiuser detection}.\hskip 1em plus 0.5em minus 0.4em\relax
  Cambridge university press, 1998.

\bibitem{shengli09}
S.~Zhang and S.-C. Liew, ``Channel coding and decoding in a relay system
  operated with physical-layer network coding,'' \emph{Selected Areas in
  Communications, IEEE Journal on}, vol.~27, no.~5, pp. 788--796, June 2009.

\bibitem{wubben10}
D.~Wubben and Y.~Lang, ``Generalized sum-product algorithm for joint channel
  decoding and physical-layer network coding in two-way relay systems,'' in
  \emph{Proc. IEEE GLOBECOM}, Dec
  2010.

\bibitem{liew2013physical}
S.~C. Liew, S.~Zhang, and L.~Lu, ``Physical-layer network coding: Tutorial,
  survey, and beyond,'' \emph{(invited paper) Physical Communication}, vol.~6,
  pp. 4--42, 2013.

\bibitem{feng2013algebraic}
C.~Feng, D.~Silva, and F.~Kschischang, ``An algebraic approach to
  physical-layer network coding,'' \emph{Information Theory, IEEE Transactions
  on}, vol.~59, no.~11, pp. 7576--7596, Nov 2013.

\bibitem{zhu2014isit}
J.~Zhu and M.~Gastpar, ``Gaussian (dirty) multiple access channels: A
  compute-and-forward perspective,'' in \emph{Proc. IEEE ISIT}, June 2014.

\bibitem{luby98}
M.~Luby, M.~Mitzenmacher, and M.~A. Shokrollahi, ``Analysis of {R}andom
  {P}rocesses via {A}nd-{O}r {T}ree {E}valuation,'' in \emph{Proc. SODA}, 1998.

\bibitem{yang14l}
S.~Yang, S.-C. Liew, L.~You, and Y.~Chen, ``Linearly-coupled fountain codes for
  network-coded multiple access,'' in \emph{Proc. IEEE ITW}, Nov 2014.

\bibitem{yang14c}
------, ``Linearly-coupled fountain codes,'' \emph{arXiv preprint
  arXiv:1410.2757}, 2014.

\bibitem{yang14bats}
S.~Yang and R.~Yeung, ``Batched sparse codes,'' \emph{Information Theory, IEEE
  Transactions on}, vol.~60, no.~9, pp. 5322--5346, Sep 2014.

\bibitem{yang14d}
S.~Yang and B.~Tang, ``From {LDPC} to chunked network codes,'' in
  \emph{Proc. IEEE ITW}, Nov 2014.

%
%
%
%
%
%

\end{thebibliography}
\end{document}